\documentclass[11pt]{amsart}
\usepackage{fullpage,amssymb,epic,eepic,epsfig,amsmath}
\usepackage{subfigure}
\usepackage{ifpdf}
\usepackage{sidecap}
\usepackage{verbatim}
\usepackage{amsaddr}
  
\theoremstyle{plain}

\newtheorem{theorem}{Theorem}
\newtheorem{proposition}{Proposition}[section]
\newtheorem{lemma}[proposition]{Lemma}
\newtheorem{observation}[proposition]{Observation}

\newtheorem{claim}[proposition]{Claim}

\theoremstyle{definition}
\newtheorem{definition}[proposition]{Definition}

\newtheorem*{lemmarec}{Lemma}

\usepackage{graphicx,color}
\usepackage{tikz}
\usepackage{xifthen}	
\usetikzlibrary{calc,positioning,decorations.pathmorphing,decorations.pathreplacing,trees,arrows,
positioning,scopes,shapes.gates.logic.US,fit}
\usepackage{./figstyle}

\usepackage[normalem]{ulem}

\begin{document}

\title[On path decompositions of $2k$-regular graphs]{On path decompositions of $2k$-regular graphs}

\author{F\'abio Botler\textsuperscript{1} and  Andrea Jim\'enez\textsuperscript{2}}

 
 \address{\textsuperscript{1}Instituto de Matem\'atica e Estat\'{\i}stica, Universidade de 
	   S\~ao Paulo \\
   	\textsuperscript{2}CIMFAV, Facultad de Ingenier\'ia, Universidad de Valpara\'iso	}


\thanks{We thank the partial support of 
FAPESP~(2013/03447-6) and
{CNPq~(477203/2012-4)}, Brazil.
The first author thanks partial support FAPESP Projects (Proc. 2014/01460-8 and \mbox{2011/08033-0})
and CNPq Project ({456792/2014-7}), Brazil.
The second author thanks the  partial  support of CONICYT/FONDECYT/POSTDOC\-TORADO~3150673
and of Nucleo Milenio Informaci\'on y Coordinaci\'on en Redes ICM/FIC RC130003, Chile.\\
Email: \textsuperscript{1}fbotler@ime.usp.br, \textsuperscript{2}andrea.jimenez@uv.cl}

\begin{abstract}
Tibor Gallai conjectured that the edge set of every connected graph $G$
on $n$ vertices can be partitioned into $\lceil n/2\rceil$ paths.
Let $\mathcal{G}_{k}$ be the class of all $2k$-regular graphs 
of girth at least $2k-2$ that admit a pair of disjoint perfect matchings.
In this work, we show that Gallai's conjecture holds in $\mathcal{G}_{k}$, for every $k \geq 3$.
Further, we prove that for every graph $G$ in $\mathcal{G}_{k}$ on $n$ vertices, there exists
a partition of its edge set into $n/2$ paths of lengths in $\{2k-1,2k,2k+1\}$. 

\smallskip
\noindent \textbf{Keywords.} Path decomposition, length-constrained, $2k$-regular graphs, 
perfect \mbox{matchings.}
\end{abstract}
\maketitle

\section{Introduction}

A \emph{decomposition}  of a graph is a set of subgraphs that partition 
its edge set. If all subgraphs are paths, then it is called a \emph{path decomposition}.
In~1966, Erd{\"{o}}s (see~\cite{Lovasz66}) raised the question: \emph{given $n > 0$, what is the minimum~$\lambda$ such that every connected graph $G$ on $n$ vertices admits a decomposition into $\lambda$ paths?}
Gallai conjectured that~$\lambda$ is at most $\lfloor (n+1)/2 \rfloor$.
Despite the best efforts~\cite{Fan05,JimenezW14+,Lovasz66,Pyber96}, the conjecture of Gallai remains widely open.
Over the years, Gallai's conjecture has been explored on particular classes of graphs.
In this paper, we focus on regular graphs. A seminal result of Lov\'asz~\cite{Lovasz66}
shows that all odd regular graphs satisfy the conjecture of Gallai.
In contrast, for even regular graphs not much is known. 
Indeed, in addition to classical results for complete graphs, it is known that 
the conjecture holds on 4-regular graphs~\cite{FKfav}.
We make progress within this direction. Let $k\geq 3$ an integer 
and $\mathcal{G}_{k}$ be the class of all $2k$-regular graphs 
of girth at least $2k-2$ that admit a pair of disjoint perfect matchings.
As usual, the \emph{girth} of a graph is the length of a shortest cycle.
We establish the following theorem.

\begin{theorem}\label{theo:submain}
Every $G\in \mathcal{G}_k$ on $n$ vertices has a decomposition into $n/2$ paths.
\end{theorem}

A closely related problem to the conjecture of Gallai that has drew great interest~\cite{BoMoWa14+,FavaronGK09,Kotzig57,th2013,ZhaiL06} 
is as follows: given a family of paths \(\mathcal{H}\), is there
a decomposition \(\mathcal{D}\) of $G$ such that 
each graph in \(\mathcal{D}\) is isomorphic to a graph in \(\mathcal{H}\)?
We give a step forward towards obtaining constrained 
path decompositions of even regular graphs.  
Actually, Theorem~\ref{theo:submain} follows from a stronger statement.
In this work, the \emph{length} of a graph refers to its number of edges.
The main contribution of this paper is the following statement.

\begin{theorem}\label{theo:main}
For every $G\in \mathcal{G}_k$ there exists a decomposition $\mathcal{D}$ into paths 
of lengths in $\{2k-1,2k,2k+1\}$. Further, in $\mathcal{D}$, the number of paths of length $2k-1$ 
is equal to the number of paths of length $2k+1$.
\end{theorem}

In the same vein, Botler et al. (Theorem~\(4.2\)~of~\cite{BoMoWa14+} and Theorem~\(3.2\)~of~\cite{BoMoWa15}) 
guaranteed existence of length-constrained path decompositions for ($2k-1$)-regular graphs.

\begin{theorem}[Botler et al.~\cite{BoMoWa15,BoMoWa14+}]\label{theo:for5}
Let $k \geq 3$ and \(G\) be a \(2k-1\)-regular graph of girth at least $2k-2$ that admits a perfect matching. 
 	Then \(G\) has a decomposition into paths of length $2k-1$.
\end{theorem}

This result is the outset of the proof of Theorem~\ref{theo:main}.

 \subsection{Organization of the paper}
The proof of Theorem~\ref{theo:main} consists of two main steps. In the
first step, developed in Section~\ref{sec:mainlemma}, we prove that one can decompose
graphs in $\mathcal{G}_k$ into paths of lengths in $\{2k-1,2k,2k+1\}$
and cycles of length $2k$.
In the second step, developed in Section~\ref{sec:theorem},
we show how to turn the decompositions obtained in the first step into the desired ones.
The last section contains the proof of a technical lemma (Lemma~\ref{lem:4}) on which our results
are based on.

 %

%

\section{Decomposing into paths and cycles of lengths in $\{2k-1,2k,2k+1\}$}\label{sec:mainlemma}

Let $N\subset \mathbb{N}$. We refer to a decomposition into paths and cycles as 
 an \emph{$N$-decomposition} if all its paths and cycles have their lengths in $N$.
The following definition introduces a type of decomposition that is fundamental to prove the existence
of $\{2k-1,2k,2k+1\}$-decompositions into paths only.

\begin{definition}\label{def:balanced}
Let $G$ be a graph and $k \geq 3$.
Let $v$ be a vertex of $G$ and $\mathcal{D}$ be a $\{2k-1,2k,2k+1\}$-decomposition of $G$.
We say that $\mathcal{D}$ is \emph{$k$-balanced} (or simply \emph{balanced}) at $v$ if 
the number of paths of $\mathcal{D}$ of 
length at most $2k$ with end vertex $v$ is at least 
the number of paths of $\mathcal{D}$ of 
length at least $2k+1$ with end vertex $v$.
Further, we say that $\mathcal{D}$ is a \emph{$k$-balanced decomposition} of $G$ 
if it is \(k\)-balanced at every vertex of $G$
and the number of paths of length $2k-1$ is equal to the number of paths of length $2k+1$.
\end{definition}

Recall that $\mathcal{G}_{k}$ denotes the class of all $2k$-regular graphs 
of girth at least $2k-2$ that admit a pair of disjoint perfect matchings.
The following is the main result of this section.

\begin{theorem}\label{theo:dec6} 
For every graph in $\mathcal{G}_{k}$ there is a $k$-balanced decomposition $\mathcal{D}$.
Moreover, the cycles in $\mathcal{D}$ have length $2k$.
\end{theorem}

Throughout this section, $G \in \mathcal{G}_{k}$ is a graph
on $n$ vertices and $M$ denotes one of its two disjoint perfect matchings.
The starting point of the proof of Theorem~\ref{theo:dec6} is the following.
By removing all edges of $M$ from $G$, we obtain 
a $(2k-1)$-regular graph $G'$ with girth at least $2k-2$ provided with a perfect matching.
Due to Theorem~\ref{theo:for5}, 
there exists a decomposition $\mathcal{P}$ of~$G'$ into paths of length $2k-1$.
The rest of the proof of Theorem~\ref{theo:dec6} relies on our ability to extend $\mathcal{P}$ to 
a $k$-balanced
decomposition of $G$. Indeed, Lemma~\ref{lem:4} characterizes, and shows a way to overcome, the obstructions 
to such extension. 
Before we present the proof of Theorem~\ref{theo:dec6}, we formulate this key lemma.

\subsection{A toolbox} \label{sub:toolbox}

A sequence of vertices and edges 
$W:=v_1 e_1 v_2 \cdots  v_t e_{t} v_{t+1}$ 
is called a \emph{trail} if
\(e_i=v_iv_{i+1} \in E(W)\)
for each \(i\in\{1,\ldots,t\}\), and $e_i \neq e_j$ if $i \neq j$. If $v_1 =v_{t+1}$, then $W$ is a \emph{closed trail}. 
A closed trail $W$ is called an \emph{$(M,\mathcal{P})$-alternating $r$-closed trail} 
(or simply an \emph{alternating closed trail})
if $W$ is the union of $r$ 
distinct paths $P_1,  \ldots, P_{r}$ of $\mathcal{P}$ and $r$ distinct matching edges $e_1,  \ldots, e_{r}$ 
of $M$, for some $r \geq 1$,  such that, for each $i \in \mathbb{Z}_r$
the matching edge $e_i=uv$ is such that $u$ is an end 
vertex of $P_{i}$ and $v$ is an end vertex of $P_{i+1}$.
We write $W$ as $P_0 e_0 \cdots P_{r-1} e_{r-1} P_0$.
Note that a \(1\)-closed trail is a cycle of length~$2k$. 
A subsequence of $W$ is called an \emph{$(M,\mathcal{P})$-alternating trail} 
(or simply an \emph{alternating trail}).

The following definition presents the structure of the main obstruction.
\begin{definition} \label{def:exceptional} Let $G \in \mathcal{G}_{3}$.
	We say that a subgraph of $G$ is \emph{exceptional}
	if it isomorphic to an $(M,\mathcal{P})$-alternating trail $PeP'e'$
	where \(P = x_0\cdots x_5\), \(P' = y_0\cdots y_5\),
	and \(e = x_5y_0\) such that
	 \(e' =  y_5y_2\),  \(y_0 = x_1\), \(y_4 = x_5\) and \(y_2 = x_3\)
	(see Figure~\ref{fig:badcase}).  
\end{definition}

\begin{SCfigure}[5][h]
 \hspace*{1.5cm}
 \begin{tikzpicture}[scale = .7]

	\node (0) [tiny black vertex] at (0,0) {};
	\node (a) [tiny black vertex] at (120:1cm) {};
	\node (b) [tiny black vertex] at (160:1cm) {};
	\node (c) [tiny black vertex] at (200:1cm) {};
	\node (d) [tiny black vertex] at (240:1cm) {};
	\node (1) [tiny black vertex] at (140:2cm) {};
	\node (2) [tiny black vertex] at (220:2cm) {};
	\node (x) [tiny black vertex] at (140:3cm) {};
	\node (y) [tiny black vertex] at (270:2cm) {};	

	\draw[line width=1pt, color=black] (x) -- (1) -- (b) -- (0) -- (c) -- (2);
	\draw[line width=1pt,color=red] (1) -- (a) -- (0) -- (d) -- (2) -- (y);
	
	\draw[line width=1pt,dashed] (1) -- (2) (y) -- (0);
	
\end{tikzpicture}\hspace*{0.3cm}\vspace*{-.6cm}
  \caption{Exceptional graph: 
  $P$ is the black path, $P'$ is the red path and the dashed edges represent the matching edges.}
\label{fig:badcase}
\end{SCfigure}
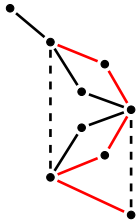
\vspace*{.3cm}

  \bigskip
The notation $s(j, l)$ stands for the $(M,\mathcal{P})$-alternating trail $P_{j} e_{j} \cdots P_{j+l} e_{j+l} P_{j+l+1}$, which embraces the matching edges \(e_j,\ldots,e_{j+l}\). 
In the case that $P_j = P_{j+l+1}$, we say that $s(j, l)$ is cyclic.
If $s(j, l)$ is non-cyclic, 
we refer to the alternating trails 
$$s[j,l) = e_{j-1}s(j, l), \quad s(j,l] = s(j, l)e_{j+l+1}, \quad s[j,l] = e_{j-1}s(j, l)e_{j+l+1},$$

\noindent
and to $s(j, l)$ itself, as \emph{extensions} of $s(j, l)$.

\begin{observation}\label{obs:cyc-exc}
Because graphs in $\mathcal{G}_{3}$ are triangle-free,
if the $(M,\mathcal{P})$-alternating trail $s(j, l)$ is cyclic and contains an exceptional
subgraph, then $l \geq 2$.
\end{observation}

The following definition captures the notion of balanced decompositions in extensions.

\begin{definition}[Quasi-balanced decompositions]
        Let $X$ be an extension of $s(j,l)$. 
        We refer to a $\{2k-1,2k,2k+1\}$-decomposition \(\mathcal{D}\) of $X$ 
        as \emph{quasi-$k$-balanced} (or simply, quasi-balanced)
        if \(\mathcal{D}\) is $k$-balanced at all vertices of \(X\),
	except possibly, at the vertices of odd degree of $X$
	incident to an edge of \(\{e_{j-1},e_{j+l+1}\}\).
	In addition, if \(\mathcal{D}\) is not $k$-balanced at 
	$v \in e$ for  $e \in \{e_{j-1},e_{j+l+1}\}$, then 
	$e \in X$.
\end{definition}

For a set $S$, the notation $|S|$ stands for the cardinality of $S$.
The next observation follows.

\begin{observation}\label{obs:quasi}
Let $X_1, \ldots, X_l$ be a partition of the edges of an $(M,\mathcal{P})$-alternating trail $s$  
(not necessarily closed) into extensions
and $\mathcal{D}_1, \ldots, \mathcal{D}_l$
be a set of quasi-\(k\)-balanced decompositions of $X_1, \ldots, X_l$, respectively.
Then, $\mathcal{D}=\cup_{i=1,\ldots,l} \mathcal{D}_i$ is a quasi-\(k\)-balanced decomposition of $s$.
Further, if for each $i \in \{1,\ldots,l\}$, $|\mathcal{D}_i|$ equals the number of paths of $\mathcal{P}$ in $X_i$
and $s$ is an alternating closed trail, then $\mathcal{D}$ is a balanced
decomposition of $s$.
\end{observation}

The toolbox of this paper is a result (Lemma~\ref{lem:4}) which claims that each alternating
trail containing 2 paths from $\mathcal{P}$ has a quasi-$k$-balanced decomposition into 2 paths,
unless it contains an exceptional subgraph (in particular, $k=3$). 
Moreover, Lemma~\ref{lem:4} shows that exceptional subgraphs are disjoint and are not subsequent,
which allows us to overcome them. 
The proof of Theorem~\ref{theo:dec6} relies on this result; specifically, 
it is used to prove Claim~\ref{niceext}.

\begin{lemma}\label{lem:4}
Let $e_{i-1}P_i \cdots P_{i+2} e_{i+2}$ be an $(M,\mathcal{P})$-alternating trail
and $X$ be an extension of $P_ie_iP_{i+1}$. 
The following two statements hold.
\begin{itemize}
 \item[\emph{(i)}] If $X$ does not contain an exceptional subgraph, then $X$ has a 
 quasi-$k$-balanced decomposition into 2 paths. 
 \item[\emph{(ii)}] If $X$ is exceptional (in particular, $k=3$),
 then $e_{i-1} \neq e_{i+1}$.
 If, in addition, $X = P_ie_iP_{i+1}e_{i+1}$ (analogously for $e_{i-1}P_ie_iP_{i+1}$), 
 then neither $e_iP_{i+1}e_{i+1}P_{i+2}$, nor $P_{i+1}e_{i+1}P_{i+2}e_{i+2}$ is exceptional. 
\end{itemize}
\end{lemma}

In the following we present the proof of Theorem~\ref{theo:dec6}.
The proof of Lemma~\ref{lem:4} is in Section~\ref{sub:techlemmas}.

\subsection{Proof of Theorem~\ref{theo:dec6}}

Recall that $G \in \mathcal{G}_{k}$ is a $2k$-regular graph on $n$ vertices 
of girth at least $2k-2$ that admits a pair of disjoint perfect matchings,
$M$ denotes one of its two disjoint perfect matchings
and the $(2k-1)$-regular graph $G'=G-M$ has a decomposition~$\mathcal{P}$ into paths of length $2k-1$.

Since each vertex of $G'$ is the end vertex of exactly
one path in $\mathcal{P}$, it follows that the edge set of $G$ admits a partition $\mathcal{W}$ 
into $(M,\mathcal{P})$-alternating closed trails.
Each alternating $1$-closed trail in $\mathcal{W}$ is a cycle of length~$2k$.
Hence, to conclude the result 
of Theorem~\ref{theo:dec6}, it suffices to prove that for each alternating $r$-closed trail with $r \geq 2$ 
there exists a $k$-balanced decomposition into $r$ paths; this is the statement of our next lemma.
Thus, the validity of Theorem~\ref{theo:dec6} follows from Lemma~\ref{lem:lemaforgpaths}.

\begin{lemma}\label{lem:lemaforgpaths}
Let $r \geq 2$ and $W$ be an $(M,\mathcal{P})$-alternating $r$-closed trail.
Then, $W$  admits a $k$-balanced decomposition.
\end{lemma}
%

Before we show Lemma~\ref{lem:lemaforgpaths}, we present two useful definitions.

\begin{definition}[malicious paths, nice paths and trapped edges]\label{def:trapped}
 Let $PeP'$ be a path-edge-path $(M,\mathcal{P})$-alternating trail.
 Suppose $e=xx'$, where $x$ is an end vertex of $P$ and 
 $x'$ is an end vertex of $P'$. 
 If $x'$ is an internal vertex of $P$, then 
 we say that $P$ is \emph{malicious} for $e$; otherwise,
 we say $P$ is \emph{nice} for $e$. If $P$ and $P'$ are malicious for $e$,
 we say that edge $e$ is \emph{trapped} (in the sequence $PeP'$).
 \end{definition}

 \begin{definition}[trapped sequences, nice and malicious extensions]\label{def:int}
Let $r \geq 2$ and $W$ be the $(M,\mathcal{P})$-alternating $r$-closed trail $P_0 e_0 \cdots P_{r-1} e_{r-1}P_0$.
If all matching edges $e_{j}, e_{j+1},\ldots , e_{j+l}$ are trapped in $s(j, l)$, 
then we say that $s(j, l)$ is a \emph{trapped sequence}.
Moreover, we say that an extension $X$ of $s(j, l)$ is \emph{nice}, if the following holds:
if $e_{j-1} \in X$ (resp. $e_{j+l+1} \in X$) and $e_{j-1} \neq e_{j+l+1}$, then $P_{j}$ (resp. $P_{j+l+1}$) 
is nice for $e_{j-1}$ (resp. $e_{j+l+1}$). Otherwise, it is called \emph{malicious}. 
In particular,  $s(j, l)$ is nice as an extension of itself.
\end{definition}

Note that an extension \(X\) of \(s(j,l)\) is malicious only if at least one of the following holds.
\begin{itemize}
	\item \(e_{j-1}\in X\), \(e_{j-1} \neq e_{j+l+1}\), 
		and \(P_j\) is malicious for \(e_{j-1}\); or
	\item \(e_{j+l+1}\in X\), \(e_{j-1}\neq e_{j+l+1}\),
		and \(P_{j+l+1}\) is malicious for \(e_{j+l+1}\).
\end{itemize}

We are ready to prove Lemma~\ref{lem:lemaforgpaths}.

\subsection*{Proof of Lemma~\ref{lem:lemaforgpaths}}

Let $W = P_0 e_0 \cdots P_{r-1} e_{r-1}P_0$ be an alternating $r$-closed trail, for some $r \geq 2$. 
Recall that we want to prove that $W$ admits a $k$-balanced decomposition. 
Note that
the property that the number of paths of length $2k-1$ equals the number of paths
of length $2k+1$ of the decomposition is equivalent to the one that the 
decomposition consists of~$r$ paths (Observation~\ref{obs:cyc-exc}). 

Let us first suppose that for every $i \in \mathbb{Z}_{r}$, the matching edge $e_i$ is not trapped in~$P_{i}e_iP_{i+1}$. 
Since $r \geq 2$, for each $i \in \mathbb{Z}_{r}$, we have $P_{i}e_i$ or $e_iP_{i+1}$ is a path. 
Moreover, if \(e_{i-1}P_{i}\) and \(P_{i}e_i\) are paths, then \(e_{i-1}P_ie_i\) is a path (because \(e_0, \ldots, e_{r-1}\) are elements of a matching of \(G\))
and it has length \(2k+1\).
Therefore, it is possible to decompose $W$ into $r$ paths whose lengths are in $\{2k-1, 2k, 2k+1\}$
such that each path of length $2k+1$ has both end edges in the perfect matching $M$. Hence,
each vertex of $W$ is the end vertex of at most one path of length $2k+1$,
thus, the decomposition is $k$-balanced.

We now assume that $W$ contains trapped edges. 
We study first the case that $W$ is not a cyclic trapped sequence,
and later the remaining case.

\subsection*{Case that $W$ is not a cyclic trapped sequence}
We say that the matching edges $e_{i-1}, e_{i}$ are the 
\emph{neighbouring edges} of $P_i$ in $W$, for all $i \in \mathbb{Z}_{r}$.
Let $\mathcal{S}$ be the set of all maximal trapped sequences in $W$,
$\mathcal{P}'\subset\mathcal{P}$ be the set of all paths that
do not belong to any trapped sequence in $W$,
and $W'$ be the subgraph of $W$ consisting of the union of all the paths in 
$\mathcal{P}'$ and all matching edges in $M$
satisfying that for each of them there exists a path in $\mathcal{P}'$ that is nice.
Note that each maximal trapped sequence starts and ends with paths in $\mathcal{P}$
and that $P \in \mathcal{P}'$ if and only if either $P$ is nice for one of its {two} neighbouring
edges, or there are paths in $\mathcal{P}$ that are nice for its neighbouring edges (and thus, such paths are also in 
$\mathcal{P}$).
Hence, $W'$ is a disjoint union of extensions.
Trivially, it follows that there exists 
a $\{2k-1,2k,2k+1\}$-decomposition  $\mathcal{P}(W')$ of $W'$ into~$|\mathcal{P}'|$ paths 
such that~$\mathcal{P}(W')$ is quasi-$k$-balanced at each extension of $W'$.

Let $W^* = W - E(W')$ be the graph obtained from $W$ by deleting the edge set $E(W')$
and all (if any) isolated vertices. Recall that $s(j, l)$ denotes $P_{j} e_{j} \cdots P_{j+l} e_{j+l} P_{j+l+1}$.
We refer to the paths $P_j$ and $P_{j+l+1}$ as the \emph{end paths} 
of $s(j, l)$, to $P_j$ as its \emph{inital path} and to $P_{j+l+1}$ as its \emph{final path}. 
Observe that $W^*$ is the disjoint union of all the maximal trapped sequences in $\mathcal{S}$
and a set $M^*\subset M$ (possibly empty). 
Let us understand the behavior of the edges in $M^*$ arise. 
If $e \in M^*$, then there is no path in $\mathcal{P}'$ that is nice for $e$.
Thus, $e$ is a neighbouring edge of an end path $P_e$ of a maximal trapped sequence $s_e$ of $\mathcal{S}$;
otherwise, there are paths $P,P'$ in $\mathcal{P}'$ such that $PeP'$ is a trapped sequence, a contradiction.
We now prove that we can choose $s_e$ and $P_e$ so that 
$P_e$ is nice for $e$.
We have already shown that there exists a maximal trapped 
sequence $s(j,l) \in \mathcal{S}$, such that $e=e_{j-1}$, or $e= e_{j+l+1}$.
{Suppose, without loss of generality, that $e=e_{j-1}$.}
Suppose that the result does not hold with such sequence; that is, 
if $e=e_{j-1}$ 
then $P_{j}$ 
is malicious for $e$.
Due to the maximality of $s(j, l)$, the path $P_{j-1}$ 
 is nice for $e$.
Hence, $P_{j-1}$ 
is neither an end path of $s(j,l)$, nor 
a path in $\mathcal{P}'$.
Therefore, $P_{j-1}$ 
is the final 
path of a maximal
trapped sequence $s \in\mathcal{S}$, obviously distinct of $s(j,l)$,  and the result follows by taking
$s_e=s$ and $P_e=P_{j-1}$ 
Consequently, we can claim the following.
%
%
%
%

\begin{claim}\label{cl:claim}
$W^*$ can be decomposed into nice extensions of maximal trapped sequences.
\end{claim}

Let us assume that the following claim holds:

\begin{claim}\label{niceext}
For every nice extension $X$
of a trapped sequence $s(j,l)$, there exists a quasi-$k$-balanced decomposition  $D$ of $X$ into $l+2$ 
paths.
\end{claim}
Then, according to Observation~\ref{obs:quasi}, Lemma~\ref{lem:lemaforgpaths} follows from the fact that $W$ is the edge disjoint union of $W'$ and $W^*$, 
the existence of a quasi-$k$-balanced decomposition for each extension of $W'$ and, 
the validity of Claims~\ref{cl:claim} and~\ref{niceext}.

In what follows, we prove Claim~\ref{niceext}.
The proof of Claim~\ref{niceext} relies on the technical result, namely Lemma~\ref{lem:4},
described in Subsection~\ref{sub:toolbox}.


%

    \begin{proof}[Proof of Claim~\ref{niceext}]
    
    Let us proceed by induction on $l$ (clearly, $l\geq 0$).
    For the case that $l=0$, we have that \(s(j,0)\) is a path-edge-path alternating trail.
    Let $X$ be a nice extension of \(s(j,0)\). By Lemma~\ref{lem:4}\emph{(i)}, if $k>3$, then $X$ has a 
    quasi-$k$-balanced decomposition
    into~2 paths.
    Suppose $k=3$.
    If $e_{j-1} \neq e_{j+1}$, then 
    $X$ does not contain exceptional subgraphs (see Figure~\ref{fig:badcase}),
    because if $e_{j-1}$ (resp. $e_{j+1}$) is in \(X\), then
    $P_{j}$ (resp. $P_{j+1}$) is nice for $e_{j-1}$ (resp. $e_{j+1}$).
    Thus, using Lemma~\ref{lem:4}\emph{(i)}
    the result follows. If $e_{j-1} = e_{j+1}$, 
    then by Observation~\ref{obs:cyc-exc}, $X$ does not
    contain an exceptional subgraph and the result holds, again, by Lemma~\ref{lem:4}\emph{(i)}.

    We now suppose $l\geq 1$.
    Let $eP_je_jP_{j+1}e_{j+1}$ be the begining of the nice extension $X(j, l)$, where $e$ may exist or not.
    If \(e\in X(j, l)\), we can suppose without loss of generality that \(P_j\) is nice for \(e\);
    otherwise, namely, if $e\in X(j, l)$ and $P_j$ is malicious for $e$, by definition of nice extensions, we have that
    \(e_{j+l+1} \in X(j, l)\) and \(P_{j+l+1}\) is nice for \(e_{j+l+1}\), and thus we can
    use the inverse sequence of \(X(j,l)\), instead of \(X(j,l)\) itself. 
    In particular, in the case that $k=3$, we have that 
    $eP_je_jP_{j+1}$ is not exceptional.

    If $k=3$ and $P_je_jP_{j+1}e_{j+1}$ is not exceptional, 
    then by Lemma~\ref{lem:4}\emph{(i)}, we have that 
    $eP_je_jP_{j+1}e_{j+1}$ has a quasi-$3$-balanced decomposition into two paths.
   For $k>3$, we have that $eP_je_jP_{j+1}e_{j+1}$ is an extension
    of $P_je_jP_{j+1}$ and thus Lemma~\ref{lem:4}\emph{(i)} claims that there exists
    a quasi-$k$-balanced decomposition of $eP_je_jP_{j+1}e_{j+1}$ into two paths.
    For $k\geq 3$, let \(X:=X(j, l)-eP_je_jP_{j+1}e_{j+1}\); the nice extension  
    of the trapped sequence \(s(j+2,l-2)\).
    If \(l = 1\), then \(X\) is simply a path of length \(2k-1\) or \(2k\).
    If \(l \geq 2\), by the induction hypothesis, \(X\) has a quasi-$k$-balanced decomposition
    into $l$ paths. Thus, by Observation~\ref{obs:quasi}, there exists
    a quasi-$k$-balanced decomposition of $X(j, l)$ into $l+2$ paths.

    We now study the case that $k=3$  and 
    $P_je_jP_{j+1}e_{j+1}$ is exceptional.
    Using Lemma~\ref{lem:4}\emph{(ii)}, we have that $e_jP_{j+1}e_{j+1}P_{j+2}e_{j+2}$
    does not contain an exceptional subgraph, and therefore, by Lemma~\ref{lem:4}\emph{(i)},
    it has a quasi-$3$-balanced decomposition into 2 paths. 
    Since $P_j$ is nice for $e$ (by previous assumption),
    we have that \(eP_j\) is a path of length \(6\) (if \(e\in X\)) or 
    simply  a path of length \(5\).
    It follows, by Observation~\ref{obs:quasi}, 
    that $eP_je_jP_{j+1}e_{j+1}P_{j+2}e_{j+2}$ can be decomposed into three paths
    that form a quasi-$3$-balanced decomposition.    
{ Analogously to the case above, (but taking \(X=X(j, l)-eP_je_jP_{j+1}e_{j+1}P_{j+2}e_{j+2}\)) we prove again that $X(j,l)$ has a quasi-3-balanced decomposition into 
     $l+2$ paths.}
    \end{proof}
%
%
%

\subsection*{Case that $W$ is a cyclic trapped sequence}    
Let $s:= s(0,r-1)$.
We show that $s$ has a quasi-$k$-balanced decomposition.
For the particular case that $r-1=1$, the result follows by 
Lemma~\ref{lem:4}\emph{(i)} and, in the case that $k=3$, in addition by Observation~\ref{obs:cyc-exc}.
We now suppose that $r-1\geq 2$.
Assume first that for $k=3$ and some \(j \in \) \(\{0,\ldots,r-1\}\),
we have that \(P_je_jP_{j+1}e_{j+1}\), \(e_{j-1}P_je_jP_{j+1}\) are not exceptionals; 
without loss of generality $j=0$.
Therefore, 
if $k=3$ and $P_{0}e_0P_{1}e_{1}$, $e_{r-1}P_{0}e_0P_{1}$ are not 
exceptionals (analogously if $k>3$), 
then by Lemma~\ref{lem:4} we can obtain a quasi-$k$-balanced decomposition of $e_{r-1}P_{0}e_0P_{1}e_{1}$ into two paths. 
In addition, we use Claim~\ref{niceext} to obtain a quasi-$k$-balanced decomposition of 
$s(j+2, l-3) = s-e_{r-1}P_{0}e_0P_{1}e_{1}$
and thus, by Observation~\ref{obs:quasi}, we have that $s$ has a $k$-balanced decomposition.
We suppose now that  $P_{j}e_jP_{j+1}e_{j+1}$ (analogously, for $e_{j+l}P_{j}e_jP_{j+1}$) 
is exceptional for $k=3$.
Then by Observation~\ref{obs:cyc-exc}, $P_{j} \neq P_{j+2}$ and by Lemma~\ref{lem:4}\emph{(ii)}, 
$e_jP_{j+1}e_{j+1}P_{j+2}e_{j+2}$
does not contain an exceptional subgraph, thus, again by Lemma~\ref{lem:4}(\emph{i}),
$e_jP_{j+1}e_{j+1}P_{j+2}e_{j+2}$ can be decomposed into two paths that form a quasi-3-balanced decomposition. 
As before, due to Claim~\ref{niceext} we have that $s-e_jP_{j+1}e_{j+1}P_{j+2}e_{j+2}$
admits a quasi-3-balanced decomposition, and thus, by Observation~\ref{obs:quasi},
we have that $s$ has a 3-balanced decomposition.

%

\section{Path-cycles into length-constrained  paths}\label{sec:theorem}

In this section, we prove Theorem~\ref{theo:main}.
Namely, we show that for every graph in $\mathcal{G}_k$ there exists
a $\{2k-1,2k,2k+1\}$-decomposition into paths such that 
the number of paths of length $2k-1$ equals the number of paths of length $2k+1$.
From now on, the length of a graph $G$ is denoted by $l(G)$.
The following lemma is useful in the proof of the upcoming results.
{In what follows, let \(C = y_0\cdots y_{2k-1}\) and \(P = v_0\cdots v_{\ell}\),
and take the indices of the vertices of $C$ modulo \(2k\).}

\begin{lemma}\label{lemma:chord}
	Let \(C\) be a cycle of length \(2k\)
	and let \(P\) be a path with its end vertices in \(V(C)\)
	such that $E(C)\cap E(P)= \emptyset$.
	Assume that \(H:=C\cup P\) has girth at least \(2k-2\).
	Then, \(l(P)\geq k-2\), and if \(y_0=v_0\),
	then \(v_\ell = y_k\).
	Moreover, if \(P\) contains a vertex of \(C\) as an internal vertex,
	then \(l(P)\geq 2k-3\).
\end{lemma}
\begin{proof}
	Let \(x,y\in V(C)\) be the end vertices of \(P\). Since \(l(C)=2k\),
	there is a path \(P'\) in \(C\) with end vertices \(x\) and \(y\),
	satisfies \(l(P')\leq k\).
	Trivially, \(P \cup P'\) contains a cycle
	and due to the assumption on the girth of $H$, we have \(l(P)+ l(P')\geq 2k-2\).
	The first part of the result follows.
Let us now suppose that  \(z\neq x,y\) is a vertex of \(C\) and an internal vertex of \(P\).
	Given distinct vertices \(u,v,w\in V(C)\), let \(C_w(u,v)\) 
	be the path in \(C\) with end vertices \(u\) and \(v\) that avoids \(w\).
	Consider the cycles \(C_1 = P(x,z) \cup C_y(x,z)\), 
	\(C_2 = P(z,y) \cup C_x(z,y)\), and \(C_3 = P \cup C_z(x,y)\).	
	Since \(H\) has girth at least \(2k-2\), we have \(l(C_i)\geq 2k-2\), for each \(i=1,2,3\).
	Note that every edge of \(P\) is contained in exactly two cycles of \(\{C_1,C_2,C_3\}\) 
	and each edge of \(C\) is contained in exactly one of these cycles.
	Then, we have 
	{\(2k + 2\,l(P) = l(C_1)+l(C_2)+l(C_3)  \geq 6k-6\)}. The lemma follows.
\end{proof}

In what follows we prove four lemmas that are used as the initial step in the proof of 
our main theorem.

\begin{lemma}\label{lemma:sccp}
	Let \(k\geq 3\).
	Let \(C\) be a cycle of length \(2k\)
	and \(P\) be a path of length \(\ell \in \{2k-1,2k,2k+1\}\)
	such that $E(C)\cap E(P)= \emptyset$, $V(C)\cap V(P)\neq \emptyset$
        and no end vertex of \(P\) is in \(V(C)\).
	If \(H:= C\cup P\) has girth at least \(2k-2\), 
	then \(H\) can be decomposed into two paths whose lengths are in \(\{2k,\ell\}\)
	so  that each path contains exactly 
	one end vertex of $P$ {as an end-vertex}. 
\end{lemma}

\begin{proof}

In what follows, the indices of the vertices in $C$ are in $\mathbb{Z}_{2k}$.
We denote by \(C(y_i,y_j)\) (resp. $P(v_i,v_j)$) the path \(y_i\cdots y_j\) (resp. \(v_i\cdots v_j\)).
Let \(x_0,\ldots,x_r\) be the vertices of \(V(C)\cap V(P)\)
in the order that they appear when reading \(P\) starting from \(v_0\).
Suppose without loss of generality that \(x_0 = y_0\).

Suppose that $r=0$; namely, $|V(C)\cap V(P)|=1$.
Let \(t\) be the length of \(P(x_0,v_{\ell})\)
and without loss of generality suppose that $t \leq \lfloor\ell/2\rfloor \leq k$.
Since no end vertex of \(P\) is in \(V(C)\),
we have \(1 \leq t \leq k\).	
It implies that the paths \(P_1 = P(v_0,x_0) \cup C(y_0,y_{t})\) and \(P_2 = C(y_t,y_0) \cup P(x_0,v_{\ell})\) 
have lengths \(\ell\) and \(2k\), respectively.
Moreover, each of \(P_1\)  and \(P_2\) contains an end vertex of $P$. 
Thus, they form the desired decomposition of $H$.

From now on, we  can assume that { \(r>0\), i.e, \(|V(C)\cap V(P) > 1\)}.	
We claim that \(r\leq 3\) and  if \(r= 3\), then \(k\leq 4\).

\begin{claim}\label{claim:kgeq5}  
	\(r\leq 3\) and  if \(r= 3\), then \(k\leq 4\).
\end{claim}
	\noindent
	\emph{Proof.}
	Assume the opposite for the first conclusion, i.e. $r\geq 4$ (resp. \(r= 3\) for the second conclusion).
	Then, there are distinct \(x_0,x_1,x_2,x_3,x_4\) in \(V(C)\cap V(P)\) (resp. \(x_0,x_1,x_2,x_3\) in \(V(C)\cap V(P)\)).
	Due to the assumption on the length of $P$ and to that no end vertex
	of $P$ is in $V(C)$, \(P(x_0,x_{4})\) (resp. \(P(x_0,x_{3})\)) has length at most $2k-1$.
	For $r\geq 4$, using Lemma~\ref{lemma:chord}, we obtain that the length of \(P(x_0,x_4)\) 
	is at least \(2(2k-3)\) and thus, $k < 3$, a contradiction.
	For $r= 3$, again using Lemma~\ref{lemma:chord}, we obtain that the length of \(P(x_0,x_4)\) 
	is at least \(3k-5\) and thus, $k \leq 4$, as desired.~$\Box$ \bigskip

We divide the rest of the proof depending on $r\in\{1,2,3\}$. \smallskip

\noindent
	\textbf{Case 1: \(r = 1\).}	
	Let \(l_1,l_2,l_3\) be the lengths of \(P(v_0,x_0),P(x_0,x_1),P(x_1,v_\ell)\),
	respectively, and let \(c_1,c_2\) be the lengths of \(C(x_0,x_1)\) and \(C(x_1,x_0)\),
	respectively.
	We suppose without loss of generality that \(l_1 \leq l_3\) and \(c_2 \leq c_1\).
	Note that \(c_1 \geq k\).
	By Lemma~\ref{lemma:chord}, we have \(l_2 \geq k-2\).
	Thus \(l_1 + l_3 \leq k+3\) and \(1 \leq l_1 \leq (k+3)/2\).
	Since \(k \geq 3\), we have \(l_1 \leq k\).
	Moreover, \(l_1 = k\) if and only if \(k=3\).
	If \(l_1 < k\), we have \(l_1 < c_1\).
	Let \(P_1 = C(y_0,y_{l_1}) \cup P(x_0,v_\ell)\), and \(P_2 = P(v_0,x_0) \cup C(y_{l_1},y_0)\).
	Note that the only vertex in common between \(C(y_0,y_{l_1})\) and \(P\) is \(x_0 = y_0\).
	Therefore, \(P_1\) is a path of length \(l_1 + l_2 + l_3 = \ell\).
	Also, the only vertex in common between \(C\) and \(P(v_0,x_0)\) is \(x_0\), 
	and since \(l_1 \geq 1\), we have that \(C(y_{l_1},x_0)\) is not a cycle.
	Therefore \(P_2\) is a path of length \(2k\).
		
	Now, suppose \(l_1 = k = 3\).
	Since \(l_1 \leq l_3\), we have \(l_3 = 3\), and \(l_2 = 1\).
	Since \(l_2 = 1\), by Lemma~\ref{lemma:chord}, we have \(c_1 = c_2 = 3\).
	Let \(P_1 = C(y_0,y_2) \cup P(x_0,v_\ell)\), and \(P_2 = P(v_0,x_0) \cup C(y_2,y_0)\).
	Clearly, \(l(P_1)=6\) and \(l(P_2)=7\) and each of \(P_1\) and \(P_2\) contains one of the end vertices of \(P\).

\noindent
	\textbf{Case 2: \(r=2\).}
	Let \(x_0 = y_0\), \(x_1 = y_i\), and \(x_2 =y_j\).
	We suppose without loss of generality that \(0 < i < j\).
	Let \(l_1,l_2,l_3,l_4\) be the lengths of \(P(v_0,x_0),P(x_0,x_1),P(x_1,x_2),P(x_2,v_\ell)\),
	and let \(c_1,c_2,c_3\) be the lengths of \(C(y_0,y_i), C(y_i,y_j),C(y_j,y_0)\){, respectively}.
	If \(c_1 > l_1\), then we make \(P_1 = C(y_0,y_{l_1}) \cup P(x_0,v_\ell)\), 
	and \(P_2 = P(v_0,x_0) \cup C(y_{l_1},y_0)\).
	Clearly, \(l(P_1)=\ell\) and \(l(P_2)=2k\).	
	Thus, we can suppose \(c_1\leq l_1\).
	Consider the cycles \(P(x_0,x_1) \cup C(y_0,y_i)\) of length \(l_2 + c_1\),
	and \(P(x_1,x_2) \cup C(y_i,y_j)\) of length \(l_3 + c_2\).
	Since the girth of \(H\) is at least \(2k-2\), we have \(l_2 + c_1 \geq 2k-2\)
	and \(l_3+c_2 \geq 2k-2\),
	from which we obtain $c_1 + c_2 + l_2 + l_3 \geq 4k-4$.
	Moreover, since \(P\) is a path of length at most \(2k+1\), 
	we have	$l_1 + l_2 + l_3 + l_4 \leq 2k+1$.
	Thus, subtracting~ this inequality from~the previous one, we obtain
$c_1 + c_2 - l_1 - l_4 \geq 2k -5.
	$
	Since \(k\geq 3\), we have \(c_1 + c_2 - l_1 - l_4 > 0\), 
	which implies \(c_2 - l_4 > l_1 - c_1\).
	Since \(l_1 \geq c_1\), we have \(c_2 > l_4\).
	Thus, the paths \(C(y_{k-l_4},y_0) \cup P(v_0,x_2)\), 
	and \(P(x_2,v_\ell) \cup C(y_j,y_{k-l_4})\)
	of lengths \(\ell\), and \(2k\), respectively form the desired decomposition.

\noindent
	\textbf{Case 3: \(r = 3\).}
	Let \(l_1, \ldots,l_5\) be the lengths of 
	\(P(v_0,x_0)\), \(P(x_0,x_1)\), \(P(x_1,x_2)\), \(P(x_2,x_3)\), 
	\(P(x_3,v_\ell)\), respectively.	
	Due to Claim~\ref{claim:kgeq5}, we can suppose \(k\leq 4\).
	First, suppose \(k = 4\).
	We state that the following hold 
	\[l_1 = l_5 = 1,\, l_2 = l_4 = 2,\, l_3 = 3.\]
	By Lemma~\ref{lemma:chord}, we have \(l_2,l_3,l_4 \geq 2\), \(l_2+l_3\geq 5\) and \(l_3 + l_4 \geq 5\).
	Therefore, we have \(l_2 + l_3 + l_4 \geq 7\).
	Since \(l_1\geq 1\),  \(l_5 \geq 1\) and \(\ell \leq 9\),
	we have \(l_1 = l_5 = 1\), and \(l_2 + l_3 + l_4 = 7\).
	Moreover, \(\ell = 9\).
	Note that if \(l_3 \leq 2\), then \(l_2 = 3\) and \(l_3 = 3\), 
	implying \(l_2+l_3+l_4 \geq 8\), a contradiction.
	Therefore, we have \(l_3 \geq 3\) and \(l_2 = l_3 = 2\).
	
	Now, since \(l_2 = 2\), we have that \(x_1 = y_4\).
	Since \(l_3 = 3\), and \(x_2 \neq x_0\), we have \(x_2 \in \{y_1,y_7\}\).
	By symmetry, we can suppose \(x_2 = y_1\).
	Since \(l_4 = 2\), we have that \(x_3 = y_5\).
	Thus, \(y_7\) is not a vertex of \(P\).
	Let \(P_1 = (P \setminus x_0v_0) \cup x_0y_7\) and \(P_2 = (C \setminus x_0y_7) \cup x_0v_0\).
	Clearly \(P_1\) and \(P_2\) are 
	two paths whose lengths are in \(\{2k,\ell\}\).
	Moreover, each of \(P_1\) and \(P_2\) contains one of the end vertices of \(P\).
	
	We now suppose \(k=3\).
	By Lemma~\ref{lemma:chord},
	we have \(l_2 + l_3 \geq 3\). 
	Therefore, \(l_2 + l_3 + l_4 \geq 4\)
	and thus, at least one of \(l_1=1\), \(l_5=1\) holds.
	Suppose, without loss of generality, that \(l_1=1\).
	We claim that \(x_1 \notin \{y_1,y_5\}\).
	In fact if \(x_1 \in \{y_1,y_5\}\), then \(l_2\geq 3\),
	otherwise \(P(x_0,x_1) \cup y_0x_0\) would be a cycle with length smaller than \(4\).
	By Lemma~\ref{lemma:chord}, we have \(l_3+l_4 \geq 3\), hence \(l_2 + l_3 + l_4 \geq 6\)
	and \(\ell\geq 8\).
	Therefore, \(x_1 \notin \{y_1,y_5\}\).
	On the other hand, if a vertex \(y\) in \(\{y_1,y_5\}\) is not a vertex of \(P\),
	then \(P_1 = (P \setminus x_0v_0) \cup x_0y\) and \(P_2 = (C \setminus x_0y) \cup x_0v_0\)
	decompose \(H\) into paths of lengths in \(\{6,\ell\}\).
	Thus, we may suppose \(y_1,y_5 \in V(P)\).
	Since \(r = 3\), we have \(\{y_1,y_5\} = \{x_2,x_3\}\).
	Since \(P(x_2,x_3) \cup C(y_5,y_1)\) induce a cycle in \(H\),
	we have that \(l_4\geq 2\).
	By Lemma~\ref{lemma:chord}, we have \(l_2 + l_3 \geq 3\).
	Since \(l_2 + l_3 + l_4 \leq 5\), we have \(l_2 + l_3 = 3\), \(l_4 = 2\),
	and \(l_1 = l_5 = 1\).
	Let \(v\) be the neighbor of \(x_3\) in \(P(x_2,x_3)\).
	Since \(l_4 = 2\), we have that \(v\) is not a vertex of \(C\).
	Let \(P_1 = (P \setminus \{v_0x_0,vx_3\})\cup x_0x_3\) and \(P_2 = (C \setminus x_0x_3) \cup \{vx_3., v_0x_0\}\).
	Clearly, \(P_1\) is a path of length \(6\), and \(P_2\) is a path of length \(\ell = 7\).
	Moreover, each of \(P_1\) and \(P_2\) contains one of the end vertices of \(P\).
\end{proof}

We now extend Lemma~\ref{lemma:sccp}.

\begin{lemma}\label{lemma:sccp6.2}
	Let \(C\) be a cycle of length \(2k\)
	and \(P\) be a path of length \(2k\)
	such that $E(C)\cap E(P)= \emptyset$ and $V(C)\cap V(P)\neq \emptyset$.
	If \(H:= C\cup P\) has girth \(2k-2\), 
	then \(H\) can be decomposed into two paths whose lengths are in \(\{2k-1,2k,2k+1\}\).
\end{lemma}

\begin{proof}
The case that no end vertex of $P$ is in $V(C)$ holds due to Lemma~\ref{lemma:sccp}. 
Suppose now that exactly one end vertex of $P$, say $z$, is in \(V(C)\).  
Let $P'$ be the path that consists of the union of $P$, 
a new vertex \(z'\) and a new edge \(zz'\). As $P'$ has length~\(2k+1\),
by Lemma~\ref{lemma:sccp}, the graph $C \cup P'$ has a decomposition
into two paths of lengths \(2k\) and \(2k+1\).
Since the edge \(zz'\) is an end edge of one of these paths, by removing
it from such path we obtain a decomposition of $H$ into two paths 
of length \(2k-1\) and \(2k+1\), or into two paths, each of length~\(2k\).

We now assume that both end vertices of $P$ are in \(V(C)\).
Moreover, we can assume that
$y_1, y_{2k-1} \in \{x_0,\cdots,x_r\}$.
To see this, we suppose, without loss of generality, that \(y_1 \notin \{x_0,\cdots,x_r\}\).
Hence \(C \setminus x_0y_1\) and \(P \cup x_0y_1\) are paths that decompose
	\(H\), whose lengths are \(2k-1\) and \(2k+1\), respectively.

%

The next claim helps to complete the proof of the lemma.

\begin{claim}\label{claim:c6p62}
If \(x \in \{y_1,y_{2k-1}\}\)
	and the neighbor \(x'\) of \(x\) in the path \(P(x_0,x)\) is not in \(V(C)\),
	then $H$ admits a decomposition into two paths, each of length \(2k\).
\end{claim}
\noindent
\emph{Proof.}
The paths \(P(x',x_0) + x_0x + P(x,v_{\ell})\)
	and \(C - x_0x + xx'\) form a $2k$-decomposition of $H$.~$\Box$

%
%


Recall that for every \(i\) in \(\{0,\ldots,r-1\}\), 
the path \(P(x_i,x_{i+1})\) has length at least \(k-2\).
Thus, if \(k\geq 4\), the path \(P(x_i,x_{i+1})\) has length at least \(2\).
Therefore, the neighbor \(x'\) of \(x_{i+1}\) in the path \(P(x_i,x_{i+1})\)
\big(and also in the path \(P(x_0,x_{i+1})\)\big)
is not a vertex of \(C\),
and by Claim~\ref{claim:c6p62}, we can decompose \(H\) into two paths of length \(2k\).

In consequence, we can assume \(k=3\).
Let us assume that $H$ does not have a $\{5,6,7\}$-decomposition 
into paths.
If \(y_1=x_i\), then by Claim~\ref{claim:c6p62},
\(P(x_{i-1},x_i)\) has length \(1\) and thus, since \(H\) is triangle-free \(x_{i-1} = y_4\).
Analogously, if \(y_5=x_j\), then \(x_{j-1} = y_2\).
Suppose without loss of generality that \(i < j\). 
We can write $$P = P(x_0,x_{i-1}) \cup P(x_{i-1},x_i) 
		\cup P(x_i,x_{j-1}) \cup P(x_{j-1},x_j) \cup P(x_j,v_{\ell}).$$
We have \(x_{i-1} = y_4\), thus \(P(x_0,x_{i-1})\) has length at least \(2\),
because $H$ is triangle-free.
In addition,  
we have that \(P(x_i,x_{j-1})\) has length at least \(3\),
because \(x_i = y_1\) and \(x_{j-1} = y_2\).
In consequence, \(P\) has length at least \(7\), a contradiction.
\end{proof}

\begin{lemma}\label{cor:sccp5}
	Let \(C\) be a cycle of length \(2k\)
	and \(P\) be a path of length \(2k-1\) such that $E(C)\cap E(P)= \emptyset$ and $V(C)\cap V(P)\neq \emptyset$.
	If \(H:= C\cup P\) has girth \(2k-2\), 
	then \(H\) can be decomposed into two paths whose lengths are in \(\{2k-1,2k\}\).
\end{lemma}

\begin{proof}
	{ Let \(z_1,z_2\) be the two end-vertices of \(P\).}
	First, we add two new vertices \(z_1'\), \(z_2'\) and two new edges \(z_1z_1',z_2z_2'\) to \(P\).
	By Lemma~\ref{lemma:sccp}, we obtain a decomposition of this new graph into two paths,
	one of length \(2k\) and the other of length \(2k+1\) such that each path contains
	one of the new vertices $z_1', z_2'$.
	By removing the edges \(z_1z_1'\) and \(z_2z_2'\) from them, 
	we obtain two paths of lengths \(2k-1\) and~\(2k\).
\end{proof}

Finally, we consider a lemma that helps decomposing the union of 2 cycles into paths.

\begin{lemma}\label{lemma:2cycles}
	Let \(C\), $C'$ be cycles of length \(2k\)
	such that $E(C)\cap E(C')= \emptyset$ and $V(C)\cap V(C') \neq \emptyset$.
	If \(H:= C\cup C'\) has girth at least \(2k-2\), 
	then \(H\) can be decomposed into two paths of length \(2k\).
\end{lemma}
\begin{proof}
Let $C=y_0 \cdots y_{2k-1}$ and $C'=z_0 \cdots z_{2k-1}$. Without loss of generality $y_0=z_0$.
First, we claim that  $y_1$ or $y_{2k-1}$ is not in $V(C)\cap V(C')$ and 
$z_1$ or $z_{2k-1}$ is not in $V(C)\cap V(C')$.
Suppose \(z_{2k-1}\) is a vertex of \(C\).
Let \(P\) be a path in \(C\) of length at most $k$ and end vertices \(y_0\), \(z_{2k-1}\).
Note that \(P \cup y_0z_{2k-1}\) is a cycle of length \(c\) at most \(k+1\) in \(H\).
Since the girth of \(H\) is at least \(2k-2\), we have
\[
	k+1 \geq c \geq 2k-2
\]
Therefore,  \(k=3\).
Now, since the girth of \(H\) is at least \(4\),
and \(C\) has length \(6\),
if \(z_i\) is a vertex of \(C\), for \(i \in \{1,5\}\),
then \(z_i = y_3\).
Therefore \(z_1 = z_5 = y_3\),
and \(y_0z_1,y_0z_5\) is a cycle of length \(2\), a contradiction.

Now, suppose without loss of generality, that $y_1$ and $z_1$ are not in $V(C)\cap V(C')$, then the graphs $C - y_0y_1 + y_0z_1$ 
and $C' - y_0z_1 + y_0y_1$ are paths of length \(2k\).
\end{proof}

{
In order to prove Theorem~\ref{theo:main}, we consider the following definition.
\begin{definition}\label{def:complete}
Let \(G\) be a graph in \(\mathcal{G}_k\),
and let \(\mathcal{L}\) be a \(\{2k-1,2k,2k+1\}\)-decomposition of \(G\) with cycles of length \(2k\)
such that the number of paths of length $2k-1$ equals the number of paths
of length $2k+1$. 
We say that \(\mathcal{L}\) is \emph{complete} if the following two conditions hold.
\begin{itemize}
	\item[(i)]	the cycles in \(\mathcal{L}\) are vertex-disjoint; and
	\item[(ii)]	if \(v\) is a vertex of a cycle in \(\mathcal{L}\),
			then \(\mathcal{L}\) is \(k\)-balanced at \(v\).
\end{itemize}
\end{definition}
}

\begin{proof}[\emph{\textbf{Proof of Theorem~\ref{theo:main}}}]
Let \(G \in \mathcal{G}_k\). 
Due to Theorem~\ref{theo:dec6}, there exists a 
\(k\)-balanced decomposition $\mathcal{L}$ of $G$ with each cycle of length \(2k\). 
Actually, because of Lemma~\ref{lemma:2cycles}, we can also assume
that the cycles in \(\mathcal{L}\) are vertex-disjoint.
In consequence, \(\mathcal{L}\) is a complete decomposition of \(G\).
%
%
Among all complete decompositions of $G$, we consider
one, say \(\mathcal{D}\), that minimizes the number of cycles;
note that \(\mathcal{D}\) is not necessarily \(k\)-balanced. 
If \(\mathcal{D}\) has no cycles, then \(\mathcal{D}\) is the desired decomposition.
Therefore, we assume that \(\mathcal{D}\) has at least one cycle. 
Let us first show the following statement.

\begin{claim}\label{claim:ccp}
	For every cycle \(C\) in \(\mathcal{D}\)
	there are at least two paths in \(\mathcal{D}\) of length~\(2k\)
	such that each of them has exactly one end vertex in \(V(C)\).
\end{claim}

\noindent
\emph{Proof.}
Let \(\mathcal{P}\) be the set of paths in \(\mathcal{D}\) that have vertices in common with \(C\).
Since the degree of each vertex in \(G\) is greater than \(5\), 
and any two cycles of \(\mathcal{D}\)
are vertex-disjoint, we have \(\mathcal{P} \neq \emptyset\); indeed, for each vertex of $C$
there are at least two paths in $\mathcal{D}$ containing such vertex.

We claim that every path in \(\mathcal{P}\)
has an end vertex in \(V(C)\).
In fact, suppose that there is a path \(P\) in \(\mathcal{P}\) 
of length \(\ell \in \{2k-1,2k,2k+1\}\) that 
has no end vertices in \(V(C)\). 
By Lemma~\ref{lemma:sccp} we can decompose \(C\cup P\)
into paths \(P_1, P_2\) of lengths \(2k, \ell\), respectively. Moreover,
each path $P_1, P_2$ contains an end vertex of $P$. Hence, 	
we have \(\mathcal{D}' = (\mathcal{D} \setminus \{C, P\}) \cup \{P_1,  P_2\}\)
is a decomposition 
of \(G\) into paths of lengths in \(\{2k-1,2k,2k+1\}\) and cycles of length \(2k\),
and the number of cycles in $\mathcal{D}'$ is strictly smaller than the number of cycles
in $\mathcal{D}$. If, in addition, \(\mathcal{D}'\) is complete, then we have a contradiction
to the minimality of $\mathcal{D}$. Indeed, 
since $\mathcal{D}$ satisfies~(i), $\mathcal{D}'$ also does.
Note that, since $\mathcal{D}$ satisfies~(ii) and $\mathcal{D}'$ is trivially
balanced at the common end vertex of $P_1, P_2$,
the only vertices where \(\mathcal{D}'\) may contradict property~(ii) are the end vertices of $P$;
let $v_1, v_2$ be the end vertices of $P$ which are end vertices of $P_1, P_2$, respectively.
Suppose \(\mathcal{D}\) is balanced at $v_1, v_2$.
If the lengths of $P_1, P_2$ are in $\{2k-1,2k\}$, 
then \(\mathcal{D}'\) is balanced at $v_1, v_2$.
If the length of  $P_2$ is $2k+1$, then the length of $P$ is \(2k+1\) as well and thus,
the number of paths of length \(2k+1\) ending at $v_2$ 
in \(\mathcal{D}'\) is the same as the number
of paths of length \(2k+1\) ending at $v_2$ in \(\mathcal{D}\). 
Hence, \(\mathcal{D}'\) satisfies~(ii).

We now claim that \(\mathcal{P}\) does not contain paths of length \(2k-1\).
In fact, if there is a \(P\) of length \(2k-1\) in \(\mathcal{P}\),
due to Lemma~\ref{cor:sccp5}, \(C\cup P\) can be decomposed into paths \(P_1\) and \(P_2\)
of lengths \(2k-1\) and \(2k\). As $\mathcal{D}$ is complete, 
\((\mathcal{D} \setminus \{C, P\}) \cup \{P_1, P_2\}\) is complete as well,
and it has 
less cycles than \(\mathcal{D}\), a contradiction.
	
Suppose that there is a path \(P\) of length \(2k\) in \(\mathcal{P}\)
such that both end vertices are in \(V(C)\).
By Lemma~\ref{lemma:sccp6.2}, 
we can decompose \(C\cup P\) into paths \(P_1\) and \(P_2\) of length in \(\{2k-1,2k,2k+1\}\).
Since both end vertices of \(P\) are in \(V(C)\), because of~(i),
they are not vertices of a cycle in the decomposition \((\mathcal{D} \setminus \{C, P\}) \cup \{P_1,  P_2\}\),
and thus, it is complete; a contradiction to the minimality of the number of cycles in $\mathcal{D}$.
Therefore, every path of length \(2k\) in \(\mathcal{P}\) has exactly one end vertex that is a vertex of \(C\).

We conclude that if \(P\) is an element of \(\mathcal{P}\),
then one of the following happens:
either	\(P\) has length \(2k+1\) and at least one end vertex of \(P\)
is in \(V(C)\), or \(P\) has length \(2k\) and exactly one end vertex of \(P\)
is in \(V(C)\).
Therefore, there exists a vertex $v$ in $V(C)$ that is an end vertex of a path in $\mathcal{P}$.
Since the degree of $v$ is even (recall it is \(2k\)) and $\mathcal{D}$ is balanced at $v$, there 
is a positive even number of paths
in $\mathcal{P}$ ending at $v$. Furthermore, at most half of them have length \(2k+1\); 
that is, at least half of them have length \(2k\). 
Thus, if there are at least four paths, the result follows. 
Therefore, we suppose that exactly two paths $P, P'$ end at $v$.
If the length of each of these paths is \(2k\), the claim holds. 
If not, one of them has length \(2k\) and the other one has length \(2k+1\). 
Since the degree of $v$ is at least \(6\), there is a path $P^*$
that contains $v$ as an internal vertex and $P^* \notin \{P,P'\}$. 
If $P^*$ has length \(2k\), the claim follows. 
If $P^*$ has length \(2k+1\), then there is an end vertex $u$ of $P^*$ in $V(C)$ and $u \neq v$. 
As $\mathcal{D}$ is balanced at $u$, there is a path  $P''$ of length \(2k\) with $u$ 
as an end vertex. 
Since $P''$ has exactly one end vertex in $V(C)$, we have $P'' \notin \{P,P'\}$ and the result follows.
$\Box$

\smallskip

Let \(P\) be a path of length \(2k\) in \(\mathcal{D}\). 
This path exists due the assumption of the existence
of cycles and Claim~\ref{claim:ccp}.
We observe that the number $N$ of cycles in \(\mathcal{D}\)
that contain an end vertex of \(P\) is in \(\{0,2\}\).
In fact, since cycles in \(\mathcal{D}\) are vertex-disjoint, 
we have \(N \leq 2\) and
if \(N = 1\), then, by Lemma~\ref{lemma:sccp6.2},
we can obtain a complete decomposition of \(G\) with less cycles than $\mathcal{D}$, a contradiction
to the minimality of $\mathcal{D}$.
	
Now, consider the auxiliary graph \(K\) with 
the set of cycles in \(\mathcal{D}\) as vertex set of \(K\),
and such that two vertices \(C_i\), \(C_j\) form an edge
if and only if there exists a path of length \(2k\) in \(\mathcal{D}\) with an end vertex in $C_i$
and an end vertex in $C_j$.
Because of the previous observation and by Claim~\ref{claim:ccp}, 
the minimum degree of \(K\) is \(2\) and thus,
\(K\) contains a cycle \(C_0C_1\cdots C_{t-1}C_0\).
For each \(i \in \mathbb{Z}_{t}\), let \(P_i\) be a path in 
\(\mathcal{D}\) such that each cycle $C_i$, $C_{i+1}$
contains an end vertex of $P_i$.
By Lemma~\ref{lemma:sccp6.2}, for each \(i \in \mathbb{Z}_{t}\),
there exists a decomposition of \(C_i\cup P_i\) into two paths of lengths in $\{2k-1,2k,2k+1\}$ 
and thus, we can obtain a decomposition of \(\bigcup_{i \in \mathbb{Z}_t} (C_i\cup P_i)\) into $2t$
paths of lengths in $\{2k-1,2k,2k+1\}$ such that the number
of paths of length $2k-1$ equals the number of paths of length $2k+1$. 
This yields a complete decomposition of \(G\)
with $t$ cycles less than \(\mathcal{D}\), 
a contradiction to the minimality of \(\mathcal{D}\).
\end{proof}

\section{Proof of Lemma~\ref{lem:4}}\label{sub:techlemmas}

We recall Lemma~\ref{lem:4}. In the following, $G \in \mathcal{G}_k$, 
$M$ is a perfect matching of $G$ and $\mathcal{P}$ is a $(2k-1)$-decomposition
of $G-M$ into paths.

\begin{lemmarec}[Lemma~\ref{lem:4}]
Let $e_{i-1}P_i \cdots P_{i+2} e_{i+2}$ be an $(M,\mathcal{P})$-alternating trail
and $X$ be an extension of $P_ie_iP_{i+1}$. 
The following two statements hold.
\begin{itemize}
 \item[\emph{(i)}] If $X$ does not contain an exceptional sequence, then $X$ has a 
 quasi-$k$-balanced decomposition into 2 paths. 
 \item[\emph{(ii)}] If $X$ is an exceptional sequence (in particular, $k=3$),
 then $e_{i-1} \neq e_{i+1}$.
 If, in addition, $X = P_ie_iP_{i+1}e_{i+1}$ (analogously for $e_{i-1}P_ie_iP_{i+1}$), 
 then neither $e_iP_{i+1}e_{i+1}P_{i+2}$, nor $P_{i+1}e_{i+1}P_{i+2}e_{i+2}$ is exceptional. 
\end{itemize}
\end{lemmarec}

Let $P$ be a path and $x$, $y$ be vertices of $P$. In what follows, the notation
$P(x,y)$ stands for the subpath of $P$ with end vertices $x$ and $y$;
we refer to it as a \emph{segment} of $P$; 
{if \(x=y\), then}
$P(x,y)$ is simply a vertex of $P$.

Recall that for a trapped sequence  $PeP'$,
it holds that \(e \subset V(P)\cap V(P')\).
The following claim follows due to the girth condition
on $2k-2$ and since \(P,P'\) do not have common end vertices.

%


\begin{claim}\label{lem:3}
Let $PeP'$ be a trapped sequence, $e{:=}xx'$
and suppose that \(|V(P)\cap V(P')| >2\).
If $k>3$, then \(|V(P)\cap V(P')| =3\),
the length of $P(x,x')$ and of $P(x,x')$ is $2k-2$
and the middle vertex of $P(x,x')$ corresponds to the middle vertex of $P(x,x')$.
If $k=3$, then
$PeP'$ is isomorphic to one of the graphs described in Figure~\ref{fig:pepconfig}. 
\end{claim}

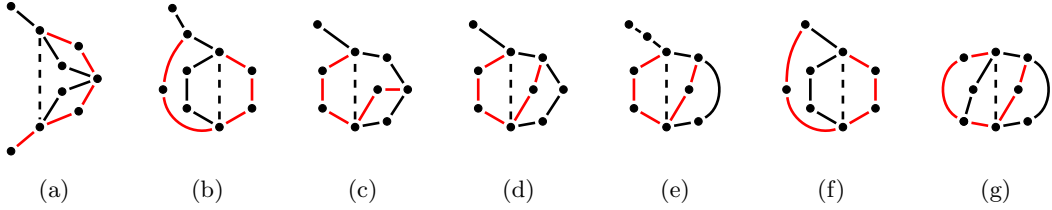
\begin{figure}[h] 
\centering 
\subfigure[]{\begin{tikzpicture}[scale = .5]

	\node (0) [tiny black vertex] at (0,0) {};
	\node (a) [tiny black vertex] at (120:1cm) {};
	\node (b) [tiny black vertex] at (160:1cm) {};
	\node (c) [tiny black vertex] at (200:1cm) {};
	\node (d) [tiny black vertex] at (240:1cm) {};
	\node (1) [tiny black vertex] at (140:2cm) {};
	\node (2) [tiny black vertex] at (220:2cm) {};
	\node (x) [tiny black vertex] at (140:3cm) {};
	\node (y) [tiny black vertex] at (220:3cm) {};	

	\draw[line width=1pt, color=black] (x) -- (1) -- (b) -- (0) -- (c) -- (2);
	\draw[line width=1pt,color=red] (1) -- (a) -- (0) -- (d) -- (2) -- (y);
	
	\draw[line width=1pt,dashed] (1) -- (2);
	
\end{tikzpicture}}\hspace{.5cm}
\subfigure[]{\begin{tikzpicture}[scale = .5]

	\node (0) [tiny black vertex] at (30:1cm) {};
	\node (1) [tiny black vertex] at (90:1cm) {};
	\node (2) [tiny black vertex] at (150:1cm) {};
	\node (3) [tiny black vertex] at (210:1cm) {};
	\node (4) [tiny black vertex] at (270:1cm) {};
	\node (5) [tiny black vertex] at (330:1cm) {};
	
	\node (x) [tiny black vertex] at (120:2.5cm) {};		
	\node (y) [tiny black vertex] at (120:1.7cm) {};	
	\node (z) [tiny black vertex] at (180:1.5cm) {};
	
	\node (w) [] at (0,-1.5) {};

	\draw[line width=1pt,color=black] (x) -- (y) -- (1) -- (2) -- (3) -- (4);
	\draw[line width=1pt,color=red]  (z)  (4) -- (5) -- (0) -- (1);
	\draw[line width=1pt,color=red] (y) to [bend right=15] (z);
	\draw[line width=1pt,color=red] (z) to [bend right=50] (4);
	
	\draw[line width=1pt,dashed] (4) -- (1);
	
\end{tikzpicture}}\hspace{.5cm}
\subfigure[]{\begin{tikzpicture}[scale = .5]

	\node (0) [tiny black vertex] at (45:1.2cm) {};
	\node (1) [tiny black vertex] at (90:1cm) {};
	\node (2) [tiny black vertex] at (150:1cm) {};
	\node (3) [tiny black vertex] at (210:1cm) {};
	\node (4) [tiny black vertex] at (270:1cm) {};
	\node (5) [tiny black vertex] at (315:1.2cm) {};
	
	\node (x) [tiny black vertex] at (120:2cm) {};		
	\node (y) [tiny black vertex] at (0:.6cm) {};	
	\node (z) [tiny black vertex] at (0:1.4cm) {};
	
	\node (w) [] at (0,-1.5) {};

	\draw[line width=1pt,color=black] (x) -- (1) -- (0) -- (z) -- (5) -- (4);
	\draw[line width=1pt,color=red]  (1) -- (2) -- (3) -- (4) -- (y) -- (z);
	
	\draw[line width=1pt,dashed] (4) -- (1);
	
\end{tikzpicture}}\hspace{.5cm}
\subfigure[]{\begin{tikzpicture}[scale = .5]

	\node (0) [tiny black vertex] at (45:1.2cm) {};
	\node (1) [tiny black vertex] at (90:1cm) {};
	\node (2) [tiny black vertex] at (150:1cm) {};
	\node (3) [tiny black vertex] at (210:1cm) {};
	\node (4) [tiny black vertex] at (270:1cm) {};
	\node (5) [tiny black vertex] at (315:1.2cm) {};
	
	\node (x) [tiny black vertex] at (120:2cm) {};		
	\node (y) [tiny black vertex] at (0:.6cm) {};	
	\node (z) [tiny black vertex] at (0:1.4cm) {};
	
	\node (w) [] at (0,-1.5) {};

	\draw[line width=1pt,color=black] (x) -- (1) -- (0) -- (z) -- (5) -- (4);
	\draw[line width=1pt,color=red]  (1) -- (2) -- (3) -- (4) -- (y) -- (0);
	
	\draw[line width=1pt,dashed] (4) -- (1);
	
\end{tikzpicture}}\hspace{.5cm}
\subfigure[]{\begin{tikzpicture}[scale = .5]

	\node (0) [tiny black vertex] at (45:1.2cm) {};
	\node (1) [tiny black vertex] at (90:1cm) {};
	\node (2) [tiny black vertex] at (150:1cm) {};
	\node (3) [tiny black vertex] at (210:1cm) {};
	\node (4) [tiny black vertex] at (270:1cm) {};
	\node (5) [tiny black vertex] at (315:1.2cm) {};
	
	\node (x) [tiny black vertex] at (120:2cm) {};		
	\node (y) [tiny black vertex] at (0:.6cm) {};	
	\node (z) [tiny black vertex] at (110:1.5) {};
	
	\node (w) [] at (0,-1.5) {};

	\draw[line width=1pt,color=black] (x) -- (z) -- (1) -- (0)  (5) -- (4);
	\draw[line width=1pt,color=red]  (1) -- (2) -- (3) -- (4) -- (y) -- (0);	
	\draw[line width=1pt,color=black] (5) to [bend right=60] (0);

	\draw[line width=1pt,dashed] (4) -- (1);
	
\end{tikzpicture}}\hspace{.5cm}
\subfigure[]{\begin{tikzpicture}[scale = .5]

	\node (0) [tiny black vertex] at (30:1cm) {};
	\node (1) [tiny black vertex] at (90:1cm) {};
	\node (2) [tiny black vertex] at (150:1cm) {};
	\node (3) [tiny black vertex] at (210:1cm) {};
	\node (4) [tiny black vertex] at (270:1cm) {};
	\node (5) [tiny black vertex] at (330:1cm) {};
	
	\node (x) [tiny black vertex] at (120:2cm) {};		
	\node (z) [tiny black vertex] at (180:1.5cm) {};
	
	\node (w) [] at (0,-1.5) {};

	\draw[line width=1pt,color=black] (x) -- (1) -- (2) -- (3) -- (4);
	\draw[line width=1pt,color=red]   (4) -- (5) -- (0) -- (1);
	\draw[line width=1pt,color=red] (x) to [bend right=15] (z);
	\draw[line width=1pt,color=red] (z) to [bend right=50] (4);
	
	\draw[line width=1pt,dashed] (4) -- (1);
	
\end{tikzpicture}}\hspace{.5cm}
\subfigure[]{\begin{tikzpicture}[scale = .5]

	\node (0) [tiny black vertex] at (45:1.2cm) {};
	\node (1) [tiny black vertex] at (90:1cm) {};
	\node (2) [tiny black vertex] at (135:1.2cm) {};
	\node (3) [tiny black vertex] at (225:1.2cm) {};
	\node (4) [tiny black vertex] at (270:1cm) {};
	\node (5) [tiny black vertex] at (315:1.2cm) {};
	
	\node (x) [tiny black vertex] at (3) {};		
	\node (y) [tiny black vertex] at (0:.6cm) {};	
	\node (z) [tiny black vertex] at (180:.6) {};
	
	\node (w) [] at (0,-1.5) {};

	\draw[line width=1pt,color=black] (x) -- (z) -- (1) -- (0)  (5) -- (4);
	\draw[line width=1pt,color=red]  (1) -- (2)  (3) -- (4) -- (y) -- (0);	
	\draw[line width=1pt,color=red] (2) to [bend right=60] (3);
	\draw[line width=1pt,color=black] (5) to [bend right=60] (0);

	\draw[line width=1pt,dashed] (4) -- (1);
	
\end{tikzpicture}}
 \caption{Dashed edges depict matching edges. One path is depicted in black, the other one in red.
 Figures~{(a), (b), (c), (d) and (e)} correspond to \(|V(P)\cap V(P')|= 3\)
 and {(f), (g)} to \(|V(P)\cap V(P')|= 4\). }
\label{fig:pepconfig}
\end{figure}


\begin{proof}[Proof of Lemma~\ref{lem:4}{(i)}]
We first set some notation.
We write $X = e_{i-1} P_ie_i P_{i+1}e_{i+1}$ as $e^*PeP'e'$
{($e^*, e'$ might not exist)}.
Let $e^*{:=}zz', e{:=}xx'$ and $e'{:=}yy'.$
{
{Further,
in the segment \(P(x,x')\), let \(\tilde{z}\) and \(\hat{z}\) be the neighbors of \(x\) and \(x'\),
respectively;
and, in the segment, \(P'(x,x')\) let \(\hat{y}\) and \(\tilde{y}\) 
be the neighbors of \(x\) and \(x'\),
respectively.}
Thus, \(X\) is the trail 
\(zz'\cdots x \tilde{z} \cdots \hat{z}x'x\hat{y}\cdots \tilde{y}x'\cdots y'y\)
and \(P\), \(P'\) are the segments of \(X\) given by \(z'\cdots x \tilde{z} \cdots \hat{z}x'\),
\(x\hat{y}\cdots \tilde{y}x'\cdots y'\), respectively (see Figure~\ref{fig:npup}).
Note that the segments $P(z',x)$, $P'(y',x')$ are paths of length 1 (edges) or 2, and the segments $P(\tilde{z},\hat{z})$, $P'(\tilde{y},\hat{y})$ 
are paths of length $2k-3$ or $2k-4$.}

We split the proof into three cases:{ (1) \(e^*\neq e', z\notin V(P)\) and \(y\notin V(P')\),
(2) \(e^*\neq e'\) and \(z\in V(P)\), and (3) \(e^* = e'\)}.
Note that if \(e^*\neq e'\), then \(|\{z,z',y,y'\}| = 4\). In particular \(z\neq y'\) and \(y\neq z'\).
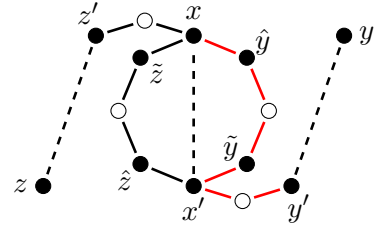
\begin{SCfigure}[4][h]
\hspace*{.5cm}
\begin{tikzpicture}[scale = 1]

	\node (0) [white vertex] at (0:1cm) {};
	\node (1) [black vertex] at (45:1cm) {};
	\node (2) [black vertex] at (90:1cm) {};
	\node (3) [black vertex] at (135:1cm) {};
	\node (4) [white vertex] at (180:1cm) {};
	\node (5) [black vertex] at (225:1cm) {};
	\node (6) [black vertex] at (270:1cm) {};
	\node (7) [black vertex] at (315:1cm) {};

	\node (z) [black vertex] at (-2,-1) {};
	\node (z') [black vertex] at (-1.3,1) {};
	\node (z'') [white vertex] at (-.65,1.2) {};
	\node (y) [black vertex] at (2,1) {};
	\node (y') [black vertex] at (1.3,-1) {};
	\node (y'') [white vertex] at (.65,-1.2) {};

	\draw[line width=1pt,color=black] (z') -- (z'') -- (2) -- (3) -- (4) -- (5) -- (6);
	\draw[line width=1pt,color=red]  (y') -- (y'') -- (6) -- (7) -- (0) -- (1) -- (2);
	
	\draw[line width=1pt,dashed] (6) -- (2) (z') -- (z) (y') -- (y);

	
	\node (1) [] at (45:1.3cm) {\(\hat{y}\)};
	\node (2) [] at (90:1.3cm) {\(x\)};
	\node (3) [] at (135:.7cm) {\(\tilde{z}\)};
	\node (5) [] at (225:1.3cm) {\(\hat{z}\)};
	\node (6) [] at (270:1.3cm) {\(x'\)};
	\node (7) [] at (315:.7cm) {\(\tilde{y}\)};

	\node (z) [] at (-2.3,-1) {\(z\)};
	\node (z') [] at (-1.4,1.3) {\(z'\)};
	\node (y) [] at (2.3,1) {\(y\)};
	\node (y') [] at (1.4,-1.3) {\(y'\)};
	
\end{tikzpicture}\hspace*{-0.4cm}\vspace{-.8cm}
 \caption{ $P$ and $P'$ are represented by black and red lines, respectively.
Dashed edges depict  matching edges.
For $k=3$, unfilled vertices may exist or not.} \label{fig:npup}
\end{SCfigure}

\vspace*{0.5cm}


\begin{itemize}
 \item[\underline{CASE 1}] $e^* \neq e'$, $z \notin V(P)$ and $y \notin V(P')$. 
 \end{itemize}
In other words, \(P\) is nice for \(e^*\) and \(P'\) is nice for \(e'\).

We first assume that $y' \neq \tilde{z}$. By the symmetry of~$X$,
this situation also covers the case that $z' \neq \tilde{y}$.
If $y \neq \tilde{z}$, then due to Claim~\ref{lem:3} the following subgraphs:   
$$P_1:= (P-x\tilde{z}) \cup e \cup  e^*  \,\,\, \text{and} \,\,\, P_2:= P' \cup x\tilde{z} \cup e'$$
form a path decomposition of $X$. Since the length of $(P-x\tilde{z}) \cup e$ is $2k-1$, 
the path $P_1$ has its length in $\{2k-1,2k\}$,
depending on whether~$e^*$ exists.
Analogously, we obtain that the path $P_2$ has its length in $\{2k,2k+1\}$.
Moreover, if $l(P_2)=2k+1$, then $e'$ belongs to $X$ and there is only
one vertex $y$ at which $\{P_1,P_2\}$ is not balanced. However, $y \in e'$ has odd degree in $X$.
Hence, $\{P_1,P_2\}$ is a quasi-$k$-balanced decomposition of~$X$.  
In the forthcoming analysis we use the same argument to guarantee that a path decomposition
is quasi-$k$-balanced, and thus, we might omit details.
Assume now that $y = \tilde{z}$. Therefore $k=3$, otherwise there is a contradiction to the girth condition
on~$X$. Due to Claim~\ref{lem:3}, we have $y' \notin V(P)$. 
Note that if $P_1\cup e'$ has length~7, then $e^*$ belongs to $X$. 
Thus, $\{P_1 \cup e', P_2 - e'\}$ is a quasi-balanced path decomposition of $X$.

Secondly, we suppose that $y' = \tilde{z}$ and $z' = \tilde{y}$.
Then, according to Claim~\ref{lem:3}, $k=3$
and $PeP'$ is isomorphic to the graph depicted in~(g) of Figure~\ref{fig:pepconfig}. 
Since $X$ is triangle-free, $z \notin V(P')$ and $y \notin V(P)$.
Let $z''$ (resp. $y''$) denote the neighbor of $z'$ (resp. $y'$) in $P$ (resp. $P'$). 
Therefore, the subgraph $P(z'',x') \cup x'z' \cup e^*$ and its complement with respect to $X$
form a quasi-balanced decomposition of $X$. 
{We clarify that along this work, the complement is edge-wise.}
That is, $K'$ is the complement of $K$ with respect to $X$ if $K'$ is obtained from $X$ by removing all edges
of $K$ and all isolated vertices of $X-E(K)$; by abuse of notation, $K' = X - K$.

\begin{itemize}
 \item[\underline{CASE 2}] $e^* \neq e'$ and $z \in V(P)$ 
\end{itemize}
In this case $k=3$. Otherwise, there is a contradiction to the condition on the girth of~$X$.
Recall that the vertex sets of $e^*$, $e$ and $e'$ are pairwise disjoint.
Trivially, if $z \in V(P)$, then~$z$ is an internal vertex of $P(x,x')$.
Since the segment $P(x,x')$ has length at least 3, there exists internal vertex of $P(x,x')$, say $n_z$,
such that $n_z$ is a neighbor of $z$.
Observe that $n_z$ is determined by the location of $z$, except in the particular case that $P(x,x')$ has length~4
and \(z\) is the middle vertex of \(P(x,x')\),
in which case $n_z \in \{\tilde{z},\hat{z}\}$.

We use the following subgraphs to split the proof into subcases:
$$H_1:=(P(x,x')-zn_z) \cup P'(x,x') \,\, \text{and} \,\, H_2: = X- H_1.$$
Note that $H_1$ and $H_2$ decompose the edge set of $X$ into two trails. 
Moreover, due to that the lengths of \(P(x,x')\) and \(P'(x,x')\) are in \(\{3,4\}\),
and \(l(H_1) = l(P(x,x')) + l(P'(x,x')) -1\), we have \(l(H_1)\in\{5,6,7\}\).

If \(H_1\) and \(H_2\) are paths and \(H_2\) has length in \(\{5,6,7\}\),
then $\{H_1,H_2\}$ is a quasi-balanced decomposition of $X$;
observe that if $H_1$ (resp. $H_2$) has length 7, then the decomposition $\{H_1,H_2\}$ is balanced
at all vertices of $X$, except at $z$ (resp. $y$), and  $z \in e^*$ (resp. $y \in e'$)
and  $e^*$ (resp. $e'$) belongs to $X$.
Note that if \(e'\in X\) and \(l(H_1)=5\), then \(l(H_2)=8\).
We assume that \(\{H_1, H_2\}\) is not a $\{5,6,7\}$-decomposition of \(X\) into paths.
Thus, one of the following situations holds:
\begin{itemize}
 \item $H_1$ is not a path.
 \item $H_2$ is not a path.
 \item the length of $H_1$ or $H_2$ is not in $\{5,6,7\}$. 
\end{itemize}

If $H_1$ is not a path, then due to Claim~\ref{lem:3}, $PeP'$ is isomorphic to the graph depicted in~(a)
of Figure~\ref{fig:pepconfig}.
Thus, $l(H_1)=7$
and $l(H_2)\in \{5,6\}$ subject to the existence of~$e'$. 
Moreover, \(P(x,x')\) and \(P'(x,x')\) have length \(4\) and \(z\in \{\hat{z},\tilde{z},z^*\}\),
where \(z^*\) is the middle vertex of \(P(x,x')\).
If \(z = z^*\), then \(X\) is exceptional{, a contradiction} (see Figure~\ref{fig:badcase}).
If \(z = \tilde{z}\), then \(\{x,z',\tilde{z}\}\) induces a triangle.
Thus, we assume that \(z = \hat{z}\) and thus, $n_z = z^*$. 

If $y\neq\hat{y}$ (resp.  $y=\hat{y}$), then  
$H_1'= H_1 - n_z\hat{y} $ and $H_2'= H_2 \cup n_z\hat{y}$ (resp. $H_1'=H_1 - n_z\tilde{z} $
 and $H_2'=H_2 \cup n_z\tilde{z}$)
form a quasi-balanced decomposition of $X$. This is because
$l(H'_1)\in\{5,6\}$, and if $l(H'_2)=7$, then $\{H'_1,H'_2\}$ is balanced at all vertices
of $X$, but at $y\in e'$ and $e'$ belongs to $X$.
Therefore, from now on, we can assume that $H_1$ is a path.

Let us discuss the second scenario; namely, $H_2$ is not a path.
We first assume that the cycles in $H_2$ arise from the intersection of the paths $P$ and $P'$. 
By definition of $H_2$, the segments of \(P\) and \(P'\) in \(H_2\) are \(P(z',x)\), \(P'(y',x')\), and \(zn_z\).
Recall that $z''$  (resp. $y''$) denotes the neighbor of $z'$ (resp. $y'$) in $P$ (resp. $P'$). 
We have  \(y' =  z''\), \(z' = y''\), or \(y' = n_z\). Further, due to Claim~\ref{lem:3}, exactly one of these equalities holds.
Note that \(z'' \neq y''\); otherwise, $X$ would have a triangle on  $\{x, x', y''\}$.
Suppose $y'=z''$. Clearly, this case is possible only if $z'' \neq x$ 
and \(y''\neq x'\).
Due to Claim~\ref{lem:3}, \(P(x,x')\) and \(P'(x,x')\) have length \(3\).
Since $X$ is triangle-free, $y \notin \{\hat{y},\tilde{z}\}$.
Let us consider the following decomposition of $X$:
$$H'_1:=e^* \cup P(z',x) \cup e  \cup P'(x',\hat{y}) \,\,\, \text{and} \,\,\, H_2': = X - H_1'.$$
If $y \neq \hat{z}$, then $\{H'_1, H'_2\}$
is a quasi-balanced decomposition of $X$. 
If $y = \hat{z}$, then \(z\neq \hat{z}\), otherwise \(z\) is incident to two matching edges, 
which implies $n_z = \hat{z}$ and $z = \tilde{z}$.
Thus, $(H'_1-e^*- z'z'')\cup e' \cup n_zz$ and its complement with respect to $X$
form a quasi-balanced decomposition.

We now suppose that $y' = n_z$.
Then, $y' \neq \hat{z}$.
Thus, \(n_z = \tilde{z}\) or \(P(x,x')\) has length~\(4\).
We claim that \(z=\hat{z}\).
In fact, if \(z = \tilde{z}\), 
then \(P(x,x')\) has length \(4\), and
\(z''=x\) and \(\{x,z',\tilde{z}\}\) induces a triangle in \(X\).
If \(z=z^*\), then the paths $H_1-z\hat{z}+ zn_z$ and 
$H_2-zn_z+ z\hat{z}$ form a quasi-balanced decomposition of $X$.

If $y \notin \{\tilde{y},\hat{y}\}$, then the path $e^* \cup P(z',n_z) \cup P'(n_z,\tilde{y})$
and its complement form a quasi-balanced decomposition of $X$.
Let $e_{x}$ denote the edge incident to $x$ in the segment $P(x,z')$.
If $y \in \{\tilde{y},\hat{y}\}$, then 
$H'_1 = e_x \cup P'(x,y) \cup e' \cup P'(y',x') \cup zx' \,\, \text{and} \,\,  
H'_2=X-H'_1$
form a $\{5,6,7\}$-decomposition 
of $X$ into paths.
Moreover, if $H'_1$ (resp. $H'_2$) has length 7, then $\{H'_1 ,H'_2\}$ is balanced
at all vertices of $X$ but at $y$ (resp. at $z$). Thus, 
$\{H'_1 ,H'_2\}$ is a quasi-balanced decomposition of $X$.

We now need to study the case that $z'=y''$. 
According to Claim~\ref{lem:3}, $PeP'$ is isomorphic to the graph depicted in~(b) or~(f)
of Figure~\ref{fig:pepconfig}.
Thus, the segments $P(x,x')$ and $P'(x,x')$  have length 3, 
and we have $z = \tilde{z}$ and $y \in \{n_z,\tilde{y},\hat{y}\}$;
note that, if \(z = \hat{z}\), 
then \(\{y'',x',\hat{z}\}\) induces a triangle in \(X\).
Moreover, $y' \neq \hat{z}$; otherwise $X$
creates a triangle or a multiple edge.
If $y\neq n_z$, then $H'_1:= e^* \cup P(z',x) \cup  P'(x,x') \cup x'\hat{z}$
and $H'_2= X - H'_1$
is a $\{5,6,7\}$-decomposition of $X$ into paths with $l(H'_1)=7$ and $l(H'_2)\in\{5,6\}$.
As $\{H'_1,H'_2\}$ is balanced at all vertices of $X$ except at $z$, the result follows.
Finally, if $y = n_z =\hat{z}$, then the paths
$$ H'_1:= P(\hat{z},z')  \cup z'x' \cup x'\tilde{y} \,\,\, \text{and} \,\,\, H'_2= X - H'_1$$
form the desired decomposition of $X$; observe that $l(H'_1)=6$, $l(H'_2)=7$ and $H'_2$ ends at~$z$.

We now move to the case that 
the cycles in $H_2$ do not arise from the intersection of the paths $P$ and $P'$. 
Hence, we necessarily have $y \in \{z'', n_z \}$. 
If \(y=z''\), then \(e^*\cup P(z',x)\cup e \cup P'(x',\hat{y})\) 
and \(X-H_1'\)
form a quasi-balanced decomposition of \(X\).
If \(y=n_z\), then \(zn_z\cup e^*\cup P(z',x)\cup e \cup P'(x',y')\)
and \(X-H_1'\) form a quasi-balanced decomposition of \(X\).
This completes the analysis for the case that $H_2$ is not a path.

Finally, we study the case that $H_1$ and $H_2$ are paths, but the length of at least
one of them is not in $\{5,6,7\}$. 
Note that, by definition, $5\leq l(H_1) \leq 7$.
On the other hand, $5\leq l(H_2) \leq 8$. 
The case $l(H_2)=8$ arises whenever $e'$ exists and 
$P(z',x)$, $P'(y',x')$ have length 2. 
In this case, if $y \neq \tilde{z}$, then the path 
$(H_1-x\tilde{z})\cup P(z',x)$ and its complement form a quasi-balanced
decomposition of $X$.
If $y = \tilde{z}$, then
the paths $H_1 \cup e^*$, of length 6, and $H_2-e^*$, of length 7,
form a quasi-balanced decomposition of $X$. The lemma follows.

\begin{itemize}
 \item[\underline{CASE 3}] $e^* = e'$  
\end{itemize}
Again, in this case $k=3$. Otherwise, there is a contradiction to the girth condition on~$X$.
Since all vertices of $X$ have even degree in $X$, every  $\{5,6,7\}$-decomposition
of $X$ into paths is balanced at every vertex of $X$.
If $PeP'$ is not isomorphic to~(a) of Figure~\ref{fig:pepconfig}, $z' \neq \tilde{y}$, and 
\(y' \neq \tilde{z}\), then, by Claim~\ref{lem:3},
we have that $P(\tilde{z},x')\cup e \cup P'(x,\tilde{y})$
and its complement form a  $\{5,6,7\}$-decomposition of $X$ into paths.
Suppose that $PeP'$ is isomorphic to~(a) of Figure~\ref{fig:pepconfig} and let $u$ be the vertex 
in $V(P)\cap V(P')-\{x,x'\}$.
Then, by Claim~\ref{lem:3}, $z' \neq \tilde{y}$ and \(y'\neq\tilde{z}\),
and thus the paths $\tilde{z}u \cup P'(u,x) \cup e \cup x'\hat{z}$ and its complement
form a $\{5,6,7\}$-decomposition of $X$.

Finally, suppose that at least one of  the following equalities $z' = \tilde{y}$, 
\(y'=\tilde{z}\),
occurs. Therefore, $z''\neq x$ and $y'' \neq x'$; since otherwise
\(\{x,x',\tilde{y}\}\) or \(\{x',y',\tilde{y}\}\) induces a triangle in~$X$.
Without loss of generality assume that $z' = \tilde{y}$ holds.
Then, $x\hat{y}\cup P(x,x')\cup x'\tilde{y}\cup \tilde{y}z''$
and its complement with respect to $X$ form the desired decomposition of~$X$.
\end{proof}
We say that a graph is an \emph{$(k,l)$-lollipop} 
if it can be obtained from the edge disjoint union of a path of length $k$, say $P$,
 and a cycle of length $l$, say $C$, satisfying that $V(P)\cap V(C)$ is an end vertex of $P$. 
\begin{proof}[Proof of Lemma~\ref{lem:4}{(ii)}]
We consider the names of the vertices of $P,P',e, e', e^*$ as in the proof of Lemma~\ref{lem:4}\emph{(i)}. 
Without loss of generality we assume that $X:=PeP'e'$.
If  $e^* = e'$, then $P'$ is not a path, a contradiction. Thus, $e^* \neq e'$.
Observe that a necessary condition on 
$eP'e'\tilde{P}$ (and $P'e'\tilde{P}\tilde{e}$) to be exceptional is 
that $P'e'$ is a $(1,5)$-lollipop, but $P'e'$ is a $(2,4)$-lollipop.
 \end{proof}

\end{document}